\documentclass{article}
\usepackage{graphicx}
\usepackage{natbib}
\usepackage{amsthm}
\usepackage{amsmath} 
\usepackage{amssymb}
\usepackage{graphicx}
\usepackage{float} 
\usepackage{hyperref}
\usepackage{color}
\usepackage{url} 
\usepackage{caption} 
\usepackage{rotating}
\usepackage{slashbox} 
\usepackage{multicol} 
\usepackage{multirow}
\usepackage{url}
\newcommand{\re}{\mathbb R}
\newtheorem{definition}{Definition}
\newtheorem{algorithm}{Algorithm}
\newtheorem{theorem}{Theorem}

\begin{document}

\title{Generalized partially linear models on Riemannian manifolds}

\author{Amelia Sim\'o, M. Victoria Ib\'a\~{n}ez, Irene Epifanio and Vicent Gimeno \\
\small{Department of Mathematics-IMAC. Universitat Jaume I.
              Avda.~del Riu Sec s/n. 12071-Castell\'o, Spain.}}

\maketitle

\begin{abstract}
The generalized partially linear models on Riemannian manifolds are introduced. These models, like ordinary generalized linear models,  are a generalization of partially linear models on Riemannian manifolds that allow for response variables with error distribution models other than a normal distribution. Partially linear models are particularly useful when some of the covariates of the model are elements of a Riemannian manifold, because the curvature of these spaces makes it difficult to define parametric models. The model was developed to address an interesting application, the prediction of children's garment fit  based on 3D scanning of their body. For this reason, we focus on logistic and ordinal models and on the important and difficult case where the Riemannian manifold is the three-dimensional case of Kendall's shape space. An experimental study with a well-known 3D database is carried out to check the goodness of the procedure. Finally it is applied to a 3D database obtained from an anthropometric survey of the Spanish child population. A comparative study with related techniques is carried out.

\textbf{keyword}Shape space;Statistical shape analysis; Generalized linear models; Partially linear models; Kernel regression;  Children's wear.
\end{abstract}

\section{Introduction} \label{introduccion}
Classification problems arise in many real-life situations. A new observation has to be classified on the basis of a training set, which is described by a set of features whose class memberships are known. Supervised learning techniques have been widely studied when the features lie on a vector space \citep{HTF09}. When features do not form a vector space, well-known supervised learning techniques are not well suited to the classifiers \cite{pedestrian}. However, features can also take values on a Riemannian manifold. This is common in  fields such as astronomy, geology, meteorology, etc., which include natural distributions on spheres, tangent bundles and Lie groups \citep{manteigaetal12}. Another example, this time in the field of computer vision, would be the space of non-singular covariance matrices \citep{pedestrian}. One  discipline that certainly offers many examples in different fields of applications  (biology, medicine, chemistry, etc.) is statistical shape analysis \citep{DrydenMardia16}. Many problems involve
predicting a categorical variable as a function of the shape of an object that lies in a Riemannian manifold.

Although different approaches can be identified in shape analysis
based on how the object is treated in mathematical terms \citep{Stoyanetal95},
the majority of research has been restricted to landmark-based analysis,
where objects are represented using $k$ labeled points
in the Euclidean space $\re^m$.
These landmarks are required to appear in each data object,
and to correspond to each other in a physical sense.
Seminal papers on this topic are
\cite{Bookstein78}, \cite{Kendall84}, and \cite{Goodall91}.
The main references are \cite{DrydenMardia16} and \cite{Kendalletal09}.
In this paper we concentrate on this approach.

In a formal way, shape can be defined as the geometrical information
about the object that is invariant under a Euclidean similarity
transformation, i.~e.,~location, orientation, and scale.
The shape space is the resulting quotient space.
When the landmark-based approach is used,
the corresponding shape space is a finite-dimensional Riemannian manifold,
and statistical methodologies on manifolds must be used.
There are several difficulties in  generalizing
probability distributions and  statistical procedures to measurements in a
non-vectorial space like a Riemannian manifold,
but fortunately, there has been a significant amount of research and activity
in this area over recent years.
An excellent review can be found in \cite{Pennec06}.

The most immediate approach for solving the classification problem when predictive variables take values on a Riemmanian manifold would be to map the manifold to a Euclidean space, i.e. to flatten the manifold. But, in a general case,  mapping that globally
preserves the distances between the points on the manifold is not available. As a consequence, the flattened space would not represent the global structure of the points appropriately. Although statistical analysis of manifold-valued data has gained a great deal of attention in recent years, there is little literature on classification. In fact, to our knowledge the only reference is \cite{pedestrian}, and it is restricted to a binary classification problem. A LogitBoost \citep{friedman00additive} on Riemannian manifolds is proposed in \cite{pedestrian}. It is similar to the original LogitBoost, except for differences at the
level of weak learners (the regression functions are learned on the tangent
space at the weighted mean of the points).

Another related work is \cite{manteigaetal12}, but they studied a regression problem (the predicted variable is real valued) rather than a classification problem. They introduced partially linear models on Riemannian manifolds (robust estimators can be found in \cite{doi:10.1080/03610926.2013.775302}). Partially linear models were proposed by \cite{Engleetal86}. Since then, partially linear models have been used in the study of complex nonlinear problems, some recent examples are \cite{ZHANG2017114}, \cite{QIAN201777}, \cite{CUI2017103} and \cite{HILAFU20171}. In this semiparametric regression method, the dependent variable is modeled with a  parametric linear part and a nonparametric part. In \cite{manteigaetal12} the variable to be non-parametrically modeled is in a manifold, and they proposed a kernel-type estimator.

Based on this idea,  we can generalize partially linear models (GPLM) to solve the classification problem with features in a Riemmanian manifold. We benefit from the flexibility of partially linear models and at the same time we can include features in a Riemmanian manifold. To our knowledge, this is the first time GPLMs have been defined on Riemmanian manifolds. At the same time, we also propose a solution for the classification problem for more than two classes, the only case studied to date. In particular, we also introduce  a solution for when the dependent variable is ordinal. Furthermore, unlike the method proposed in \cite{pedestrian} where features in a Riemmanian manifold were the only predictive variables, other predictive variables together with those  in a Riemmanian manifold are managed jointly by our proposal.

This paper addresses an important current application: size fitting for online garment shops, in particular children's
garment size matching. Customers face a challenge when they have to choose the right size of garment without try it on
 when buying these items both in
store and, especially, in online clothes shops \citep{Ding}.  Although users can base their decision on their previous experience (their virtual closet), children are constantly growing, so this not suitable strategy \citep{sindicich}. Not only that, but  each company also has its own sizing, and what is more, this can change over time \citep{nancy}. As a consequence, size matching in children should be based on their current form.

There is usually a sizing chart that corresponds to several anthropometric measurements, together with their ranges to show the size assignation. Nevertheless, customer's measurements can lead to different size assignations depending on which measurements are considered. Therefore,  customers cannot know which size will fit them best \citep{labat}. As a result, size fitting problems lead to  a  high percentage
of returns, which represents one of the main costs of this sales channel for distributors and manufacturers. The return rates of some e-commerce businesses are between 20 and 50\% \citep{showroom}. This also decreases customer satisfaction \citep{otieno} and the likelihood that the customer will buy again.  Moreover, concern about poor fit is the main obstacle to purchasing clothes online.

To address the child garment size matching problem, a fit assessment study was carried out
by the Biomechanics Institute of
Valencia. A sample of Spanish children aged between 3 and 12 years were scanned using the Vitus Smart 3D body scanner from Human
Solutions. This has  a non-intrusive laser system
consisting of four columns that house the optic system, which moves from head
to feet in ten seconds, performing a sweep of the body. The body shape of each child in our data set was represented by $3075$ 3D landmarks. Although a 3D body scanner is not usually available for  customers,  nowadays customers
can obtain their detailed body shapes using their own digital cameras or other measuring technologies \citep{Cordier,ballester}. Recently, 3D bodies have been reconstructed from images captured with a smartphone or tablet in \cite{Ballester16}.
 Furthermore, a subsample of these children tested different garments of different sizes, and their fit was assessed by an expert. This expert labeled the fit as 2 (correct), 1 (if the garment was small for the child) or 3 (if the garment was large for the child) in an ordered factor called Size-fit.

Therefore, finding the garment size that best fits  the user is a statistical classification problem \citep{meunier}. In this problem, the children's body shapes represented by landmarks are predictive variables in a Riemmanian manifold. The proposed method has been applied to the aforementioned  database of children with excellent results.

To our knowledge, the only  previous reference about the child garment size matching problem is \cite{Pierolaetal16}. However, they used multivariate features, not the complete information about the child's form. In particular, they used  the differences between the reference mannequin of the evaluated size and the child for several anthropometric measurements. If the reference mannequin is not available, that methodology cannot be used.

As regards other works that also use variables in a Riemmanian manifold in the context of the apparel industry, in \cite{Vinueetal16} women's body shapes represented by landmarks were used to define a new sizing system by adapting clustering algorithms to the shape space. Unlike our supervised learning problem, they dealt with an unsupervised learning problem. 
Another unsupervised learning problem is faced in \cite{EpiIbaSim17}, where archetypal shapes of children are discovered.


The R language \citep{R17} was employed in our implementations.
We used the \texttt{shapes} package by Ian Dryden \citep{shapesR}.
This is a very powerful and complete package for the statistical analysis
of shapes.

The article is organized as follows: Section~\ref{Sec.PLM} reviews the partially linear models on Riemmanian manifolds, which are generalized in Section~\ref{Sec.GPLM} to Riemmanian manifolds. Algorithms for their estimation are also given. Their R (\cite{R17}) code is available at \url{www3.uji.es/~epifanio/RESEARCH/partly.rar}.
Section~\ref{shape_space} describes the basic concepts of statistical shape analysis, and explains how to estimate generalized partially linear models on the Kendall's 3D  Shape Space.
 The use of the logistic partly linear model on the Kendall's 3D Shape Space is illustrated by a well-known data set in Section~\ref{simul_study}, while the ordered partially linear model is applied to solve the child garment size matching problem in Section~\ref{our_appl}.
Finally, conclusions are discussed in Section~\ref{conclusions}.


\section{Partially linear models on Riemannian manifolds } \label{Sec.PLM}
Partially linear models (PLM) \cite{Engleetal86} are regression models in which the response
depends on some covariates linearly but on other covariates nonparametrically.
PLMs generalize standard linear regression techniques and are special
cases of additive models \citep{Hastieetal90, Stone85}, which makes it easier to interpret  the
effect of each variable.

Partially linear models when one of the predictive variables takes values on a Riemannian manifold were introduced in \cite{manteigaetal12}. In this work, they consider a sample
$\{(y_i,x_i^t,s_i)\}_{i \in 1,\cdots,n}$,
where the response variable, $Y$, is a real valued scalar variable,  $x_i^t$ is a real valued $p$-dimensional vector and $s_i$ is a point of a Riemannian manifold, $M$, of dimension $d$.
They assume the partially linear model:

\begin{equation} \label{manteiga1}
y_i=x_i^t \beta + g(s_i) + \epsilon_i,  \ \ i=1, \ldots, n,
\end{equation}
and
\begin{equation} \label{manteiga2}
x_{ij}=\phi_j(s_i) + \eta_{ij},  \ \ i=1, \ldots, n, \ \ j=1, \ldots, p.
\end{equation}
with $g(s)=\phi_0(s)-\phi^t(s) \beta$, where $\phi_0(s)=E(Y \mid s)$ and $\phi(s)=(\phi_1(s), \ldots, \phi_p(s))$; and with independent errors $\epsilon_i$ and  $\eta_{ij}$. Therefore, $\beta$, $\phi_0(s)$ and  $\phi(s)$ are the parameters to estimate.

Manteiga et al. \cite{manteigaetal12} suggest estimating $\phi_0(s)$ and $\phi(s)$ using non-parametric kernel-type estimators on Riemannian manifolds (see section \ref{non_parametric_manifolds}) and then  estimating the parameter $\beta$ considering the least-squares estimator obtained by minimizing:
$$
\hat{\beta}= \text{arg} \min_{\beta} \sum_{i=1}^n \left[(y_i-\hat{\phi}_0(s_i))-(x_i-\hat{\phi}(s_i))^t \beta \right]^2.
$$
Finally, $\hat{g}(s)=\hat{\phi}_0(s)-\hat{\phi}^t (s)\hat{\beta}$.

\subsection{Non-parametric estimators on Riemannian manifolds} \label{non_parametric_manifolds}

Let $\{(x_1, s_1),\ldots,(x_n, s_n)\}$ be iid random vectors that take values on $\mathbb{R}\times M$.
Due to the curvature of $M$, kernel-type estimators of  $\phi(s)=E(x \mid s)$ must be adapted to this space.

Pelletier \cite{pelletier2006non} proposes the following non parametric estimator:
\begin{equation} \label{kernel}
\hat{\phi}(s)=\frac{\sum_{i=1}^n x_i \theta_{s}(s_i)^{-1} K_{h_n}(\rho(s-s_i))}{\sum_{i=1}^n \theta_{s}(s_i)^{-1} K_{h_n}(\rho(s-s_i))},
\end{equation}
where $ \theta_{s}(s_i)$ is the volume density function of $M$; $\rho$ is the Riemannian distance on $M$ and $K_{h_n}$ is a univariate kernel function with bandwidth $h_n$ with $\lim_{n\rightarrow\infty}h_n=0$ and $h_n<i_M$,  $i_M$ being the injectivity radius of $M$. In \cite{pelletier2006non} we can find some good properties of this estimator.

\section{Generalized Partially Linear Model on Riemannian manifolds}\label{Sec.GPLM}

As stated in the introduction, the aim of this paper is to generalize the partially linear model on Riemannian manifolds to the generalized
linear model introduced by \cite{NelderWedderburn} and to apply it to the particular and important case of the  Kendall's 3D shape space.


Although our proposal can be extended to  generalized linear models in general, we will focus on two particular important models that we will use in our applications: logistic and ordered logistic  models.

\subsection{Logistic Partially Linear Model on Riemannian manifolds}
Let $\{(y_1,x_1,s_1),\ldots,(y_n,x_n, s_n)\}$ be a set,
where $y_i$ are binary variables, $x_i$ real valued $p$-dimensional vectors and $s_i$ are points in $M$, a Riemannian manifold of dimension $d$.

Defining $p_i=E(y_i \mid x_i,s_i)$, we can assume the logistic partially linear model:

\begin{equation} \label{pglm1}
logit(p_i)=x_i^t \beta + g(s_i)   \ \ i=1, \ldots, n,
\end{equation}
and
\begin{equation} \label{pglm2}
x_{ij}=\phi_j(s_i) + \eta_{ij},  \ \ i=1, \ldots, n, \ \ j=1, \ldots, p,
\end{equation}
with $g(s)=\phi_0(s)-\phi^t(s) \beta$; $\phi(s)=(\phi_1(s), \ldots, \phi_p(s))$; and where $\beta$, $\phi_0(s)$ and  $\phi(s)$ are the parameters to estimate.

As in \cite{manteigaetal12}, because $s$ is in a Riemannian manifold, the estimation of $\phi_j(s)$ $j=1,...,p$ must be obtained using equation \ref{kernel}.
In the next section the expression of this estimator will be given for the particular and difficult case of Kendall's 3D shape space.

The algorithm that we propose follows the ideas of additive generalized linear models \citep{Hastieetal90,HTF09}: a partially linear model is applied to the adjusted dependent variable  at each step of the  iteratively reweighted
least squares (IRLS) algorithm. It is as follows (the superindex ($j$) indicates the estimation in the $j$-th iteration):

\begin{algorithm} \label{algorithm_logistic}

\begin{itemize}
\item[] $x=(x_i)_{i=1,..,n}$
\item[] Initialize $\beta^{(0)}$, $\phi_0^{(0)}=(\phi_0^{(0)}(s_1),..., \phi_0^{(0)}(s_n))^t$, $e^{(0)}$, $j=0$
\item[] Calculate $\phi=\left(
                          \begin{array}{c}
                            \phi(s_1) \\
                            \vdots \\
                            \phi(s_n) \\
                          \end{array}
                        \right)$
using equation \ref{kernel}
\item[] While ($e^{(j)}>$ specified threshold) and ($j<$ specified maximum number of steps) do
  \begin{itemize}
  \item[] Calculate $g^{(j)}=\left(
                          \begin{array}{c}
                            g^{(j)}(s_1) \\
                            \vdots \\
                            g^{(j)}(s_n) \\
                          \end{array}
                        \right)= \phi_0^{(j)}- \phi \beta^{(j)}$
  \item[] For $i=1,\ldots,n$
    \begin{itemize}
    \item[] $p_i =logit^{-1}(x_i^t \beta^{(j)} + g^{(j)}(s_i) )$
    \item[] Construct the working target variable $$z_i=x_i^t \beta^{(j)}+\phi_0^{(j)}(s_i)-\beta^{(j)}\phi^{(j)}(s_i)+\frac{y_i-p_i}{p_i(1-p_i)}$$
    \item[] $w_i=p_i(1-p_i)$
    \end{itemize}
  \item[] End for
  \item[] Apply partly linear model to the targets  $z=(z_i)_{i=1,\ldots,n}$ with weight matrix $W=Diag((w_i)_{i=1,..,n})$:
      \begin{itemize}
      \item[] Calculate $\phi_0^{(j+1)}=\left(
                                            \begin{array}{c}
                                              \phi_0^{(j+1)}(s_1) \\
                                              \vdots \\
                                              \phi_0^{(j+1)}(s_n) \\
                                            \end{array}
                                          \right)$ using equation \ref{kernel} replacing $x_i$ by $z_i$

       \item[] $\beta^{(j+1)}=((x-\phi)^tW(x-\phi))^{-1}((x-\phi)^t W(z-\phi_0^{(j+1)})) $
       \end{itemize}
 \item[] $e^{(j)}=\|\beta^{(j+1)}-\beta^{(j)}\|/\|\beta^{(j+1)}\|$
 \item[] $j=j+1$
 \item[] End while
 \end{itemize}
\end{itemize}
\end{algorithm}

With respect to the initializations, $\beta^{(0)}=0$ and $\phi_0^{(0)}=(-0.5,..., -0.5)^t$, which would correspond to equiprobability, provided good results in our experiments.

\subsection{Ordered Partially Linear Model on Riemannian manifolds}

The above algorithm can be modified to model an ordinal response, in particular we will assume the cumulative logistic model or proportional odds model \cite{mccullagh1980regression,agresti2010analysis}.

Let $\{(y_1,x_1,s_1),\ldots,(y_n,x_n, s_n)\}$ be a set,
where $y_i$ are response variables, $x_i$ real valued $p$-dimensional vectors and $s_i$ are points in $M$, a Riemannian manifold of dimension $d$. Suppose that the response variable $y$ has $K$ ordered categories and $\pi_k(x,s)=P(y\leq k \mid x,s)$, $k=1,\ldots,K-1$. Assume:
\begin{equation} \label{pglm1ordered}
logit(\pi_k(x_i,s_i)))=x_i^t \beta + g_k(s_i)   \ \ i=1, \ldots, n, \ \ k=1,\ldots,K-1
\end{equation}
and
\begin{equation} \label{pglm2ordered}
x_{ij}=\phi_j(s_i) + \eta_{ij},  \ \ i=1, \ldots, n, \ \ j=1, \ldots, p.
\end{equation}

Following   \cite{walker1967estimation,mccullagh1980regression} and \cite{thompson1981composite}, we treat the cumulative link
model as a multivariate generalized linear model \cite{fahrmeir2013multivariate} defining $Y_i=(Y_{i1},...,Y_{i(K-1)})$ as $Y_{ik}=1$ if $y_i\leq k$ and otherwise as zero.
In the multivariate case one merely has to substitute vectors and matrices for the multivariate versions.

We define the total design matrix
$\tilde{x}_i=\left(\begin{array}{c}
x_{i}^t\\
 \vdots\\
 x_{i}^t
\end{array} \right)_{(K-1) \times p}$;
$p_i=\left(\begin{array}{c} \pi_1(x_i,s_i) \\ \vdots \\ \pi_{K-1}(x_i,s_i)\end{array} \right)_{(K-1) \times 1}$
and $\tilde{\phi}(s_i)=\left(\begin{array}{c}
\phi(s_{i})^t\\
 \vdots\\
 \phi(s_{i})^t
\end{array} \right)_{(K-1) \times p}$

Let $D_i$ be the derivative of the link function and the weight matrix $W_i=D_i\sigma_i^{-1}D_i^t$, with $\sigma_i=cov(Y_i)$ (which can be considered an approximation of the inverse of the covariance matrix of the transformed response).

Algorithm \ref{algorithm_logistic} is modified as follows:

\begin{algorithm} \label{algorithm_ordinal}
\begin{itemize}
\item[] $\tilde{x}=(\tilde{x}_i)_{i=1,..,n}$
\item[] Initialize $\beta^{(0)}$, $\phi_0^{(0)}=(\phi_0^{(0)}(s_1),...,\phi_0^{(0)}(s_n))^t$, $e^{(0)}$, $j=0$
\item[] Calculate $\phi=\left(
                          \begin{array}{c}
                            \phi(s_1) \\
                            \vdots \\
                            \phi(s_n) \\
                          \end{array}
                        \right)$
 using equation \ref{kernel}
\item[] Calculate $\tilde{\phi}(s_i)$ $i=1, \ldots,n$
\item[] While ($e^{(j)}>$ specified threshold) and ($j<$ specified maximum number of steps) do
  \begin{itemize}
  \item[] Calculate $g^{(j)}=\left(
                          \begin{array}{c}
                            g^{(j)}(s_1) \\
                            \vdots \\
                            g^{(j)}(s_n) \\
                          \end{array}
                        \right)= \phi_0^{(j)}- \phi \beta^{(j)}$
  \item[] For $i=1,\ldots,n$
       \begin{itemize}
       \item[] $p_i^{(j)}=logit^{-1}(\tilde{x}_i^t \beta^{(j)} + g^{(j)}(s_i) )$
       \item[] Calculate $D^{-1}_i(p_i^{(j)})$ and $W_i(p_i^{(j)})$
       \item[] Construct the working target variable $$z_i=\tilde{x_i}^t \beta^{(j)}+\phi_0^{(j)}(s_i)-\beta^{(j)}\tilde{\phi}^{(j)}(s_i)+(D_i^{-1})^t(y_i-p_i^{(j)})$$
       \end{itemize}
  \item[] Apply partly linear model to the targets $z=(z_i)_{i=1,\ldots,n}$ with weight matrix $W=Diag((W_i)_{i=1,..,n})$ to $z=(z_i)_{i=1,\ldots,n}$:
      \begin{itemize}

        \item[] Calculate $\phi_0^{(j+1)}=\left(
                                            \begin{array}{c}
                                              \phi_0^{(j+1)}(s_1) \\
                                              \vdots \\
                                              \phi_0^{(j+1)}(s_n) \\
                                            \end{array}
                                          \right)$ using equation \ref{kernel} replacing $x_i$ by $z_i$
         \item[] $\beta^{(j+1)}=((\tilde{x}-\tilde{\phi})^tW(\tilde{x}-\tilde{\phi}))^{-1}((\tilde{x}-\tilde{\phi})^t W(z-\phi_0^{(j+1)}))$
      \end{itemize}
 \item[] $e^{(j)}=\|\beta^{(j+1)}-\beta^{(j)}\|/\|\beta^{(j+1)}\|$
 \item[] $j=j+1$
 \item[] End while
 \end{itemize}
\end{itemize}
\end{algorithm}

In the particular case of the ordinal model with three categories of our application:
$$ W_i(p_i)=\frac{1}{\pi_2(x_i,s_i)}\left(\begin{array}{cc} \frac{1-\pi_3(x_i,s_i)}{\pi_1(x_i,s_i)} & -1 \\
 -1& \frac{1-\pi_1(x_i,s_i)}{\pi_3(x_i,s_i)}
 \end{array}
 \right),$$

 $$D^{-1}_i(p_i)=\left(\begin{array}{cc} \pi_1(x_i,s_i)(1-\pi_1(x_i,s_i)) & 0\\
 0 & \pi_3(x_i,s_i)(1-\pi_3(x_i,s_i))
 \end{array}
 \right).$$

\section{Kendall's 3D  Shape Space}\label{shape_space}

In the previous section the logistic and ordered logistic  partially linear models were given for a general Riemannian manifold. In this section we give the expressions that we need in order to apply them in the particular and important case of the Kendall 's 3D Shape Space.
This manifold has a complicated structure and the calculus of the expressions needed in equation \ref{kernel}  is not trivial.

We begin by introducing some basic concepts, a complete introduction to which can be found in \cite{DrydenMardia16}.

In our approach to shape analysis each object is identified by a set of landmarks, i.e. a set of points in the Euclidean space $\mathbb{R}^{m}$ that identifies each object and match between and within populations.

\begin{definition}
A configuration matrix $X$ is a $k \times m$ matrix with the Cartesian coordinates of the $k$ landmarks of an object.
\end{definition}

The shape of an object is all the geometric information that remains invariant with translations, rotations and changes of scale. Thus:
\begin{definition}
The shape space  $\Sigma^k_m$ is the set of equivalence classes $T_X$ of $k \times m$ configuration matrices $X \in \mathbb{R}^{k \times m}$ under the action of
Euclidean similarity transformations.
\end{definition}

As  mentioned above,
the shape space $\Sigma^k_m$  admits a Riemannian manifold structure.
The complexity of this Riemannian structure depends on $k$ and $m$. For example, $\Sigma^k_2$ is the well-known complex projective space. For $m > 2$, which is the case of our application, they are not familiar
spaces and may have singularities.

A representative of each equivalence class $T_X$ can be obtained by
removing the similarity transformations one at a time. There are
different ways to do that.

Let $X$ be a configuration matrix.
One way to remove
the location effect consists of multiplying it by the Helmert submatrix,
$H$, i.~e.,~$X_H=HX$.

To filter scale, we can divide $X_H$ by the centroid size,
which is given by

\begin{eqnarray}\label{tamañocentroide}
CS(X)=\Vert X_H \Vert=\Vert HX \Vert =\sqrt{\mbox{trace}((HX)^t (HX))}=\Vert X \Vert,
\end{eqnarray}
$\Vert\cdot\Vert$ is the Frobenius norm.

So,
\begin{eqnarray}\label{preshapes}
Z_X=\frac{X_H}{\|X_H\|}
\end{eqnarray}
 is called the pre-shape of the configuration
matrix $X$ because all information about location and scale is removed,
but rotation information remains.

\begin{definition}
The pre-shape space  $S^k_m$ is  the set of all possible pre-shapes.
\end{definition}
 $S^k_m$ is a hypersphere of unit radius in $\re^{m(k-1)}$
 (a Riemannian manifold that is widely studied and
known). $\Sigma^k_m$ is the quotient space of $S^k_m$ under rotations.

As a result, a shape $S_X$ is an orbit associated
with the action of the rotation group $SO(m)$ on the pre-shape.

From now on, in order to simplify the notation, we will use $S_X$ to denote both, a configuration matrix and its shape, provided that it is understood by context.

For $m=2$, this quotient space is isometric with the complex projective
space $\mathbb{CP}_{k-2}$, a familiar Riemannian manifold without singularities.
For $m>2$, $\Sigma^k_m$ is not a familiar
space, and it has singularities.
The singularities are  shapes whose preshapes have rank $m-2$ or less. With real world applications we can usually assume that our data are almost certainly in the non-singular part of the shape space and, fortunately,
the Riemannian structure of the non-singular part of $\Sigma^k_m$ can be obtained taking into account that
$\pi: S^k_m\to \Sigma_m^k$ is a Riemannian submersion \citep{Kendalletal09}, for any $\pi(Z)\in \Sigma_m^k$ the tangent space $T_{\pi(Z)}\Sigma_m^k$ can be identified with the horizontal space $\mathcal{H}_Z$ of $T_ZS^k_m$.


\subsection{Riemannian distance}

The induced Riemannian distance in the shape space is given by the Procrustes distance defined as follows.
\begin{definition} \label{distancia_Procrustes}
Given  two configuration matrices $X_1, X_2$, the  Procrustes distance of its corresponding shapes,
 $\rho(S_{X_{1}},S_{X_{2}})$, is the closest great-circle distance  between $Z_1$ and $Z_2$ on the pre-shape
hypersphere $S^k_m$, where $Z_j=\frac{HX_j}{\|HX_j\|}, j=1,2$. The minimization is carried out over rotations.
\end{definition}

The solution for this optimization problem is:

$$\rho(S_{X_{1}},S_{X_{2}})=\arcsin\left(\sqrt{1-(\sum_{i=1}^m \lambda_i)^2}\right),$$
where $\lambda_1 \geq \lambda_2 \geq \ldots \lambda_{m-1}\geq \mid \lambda_m \mid $ are the square roots of the eigenvalues of $Z_1^TZ_2Z_2^TZ_1$, and the smallest value $\lambda_m$ is the negative square root if and only if $det(Z_1^TZ_2)<0$ \cite{DrydenMardia16}.

Note that the range of this distance is $[0,\pi/2]$.


\subsection{Volume density function}

The volume density function can be obtained taking into account that the mapping that assigns  the corresponding element on the shape space to each preshape $Y$:
\begin{eqnarray*}
&&\pi:S^k_m \rightarrow  \Sigma^k_m \\
&& Z \mapsto T=\pi(Z),
\end{eqnarray*}
is a Riemannian submersion. Then the volume density function is (see  \ref{Appendix})

\begin{equation} \label{vdf}
\theta_{\pi(Z_1)}\left(\pi(Z_2)\right)=\left(\frac{\sin{\rho(\pi(Z_1),\pi(Z_2))}}{\rho(\pi(Z_1),\pi(Z_2)))}\right)^{m(k-1)-2-\frac{m(m-1)}{2}}
\end{equation}
when $\pi(Z_1)\neq \pi(Z_2)$, and $\theta_{\pi(Z_1)}(\pi(Z_2))=1$ if $\pi(Z_1)= \pi(Z_2)$. We must stress here the importance of the volume density function  in the case of a large number of landmarks $k$, because in the limit $k\to\infty$, the definition formula (\ref{vdf}) becomes
$$
\theta_{\pi(Z_1)}\left(\pi(Z_2)\right)=\left\{\begin{array}{ccc}
1&{\rm if}&\pi(Z_1)= \pi(Z_2)\\
0&{\rm if}&\pi(Z_1)\neq \pi(Z_2)
\end{array}\right.
$$

\subsection{Generalized Partially Linear Models on the  Kendall's 3D Shape Space}

Once the necessary concepts have been introduced, we turn to the algorithm to fit a generalized partially linear model on  Kendall's Shape Space. We focus on the particular case of the ordered partially linear model (for the logistic partially linear model, it is analogous but instead we apply  the algorithm \ref{algorithm_logistic}).

\begin{algorithm} \label{algorithm_ordered_Kendall}
Given a sample
$\{(y_1,X_1,x_1),\ldots,(y_n, X_n,x_n)\}$,
where $y_i$ is a realization of an ordered variable with $K$ categories, $X_i$ are configuration matrices and $x_i$ real valued $p$-dimensional vectors.

\begin{itemize}

\item[(i)] Compute the pre-shapes of $X_1,\ldots,X_n \rightarrow Z_1,\ldots,Z_n$ using equations \ref{tamañocentroide} and \ref{preshapes}.
\item[(ii)] Apply algorithm \ref{algorithm_ordinal} with $s_i=\pi(Z_i)$ i.e. $\theta_i(s_j)=\theta_{\pi(Z_i)}\left(\pi(Z_j)\right)$ in equation \ref{vdf} and $\rho(\pi(Z_i),\pi(Z_j))$ is given by definition \ref{distancia_Procrustes}.

\end{itemize}

\end{algorithm}

\section{Application to anatomical data} \label{simul_study}

In an investigation into sex differences in the crania of a species of macaque, random samples of 9 male and 9 female skulls were
obtained by Paul O'Higgins (Hull-York Medical School) \citep{DrydenMardia16,DrydenMardia93}. A subset of seven anatomical landmarks was located on each cranium and the three-dimensional
(3D) coordinates of each point were recorded. The aim of the study was to assess whether there were any size and shape differences between
sexes. The data are available in the R \texttt{shapes} package \citep{shapesR}.

Fig. \ref{landmarks.macaque} (a) shows the distribution of all the landmarks of the 18 subjects.
From the configuration matrices $\{X_i \}_{i=1,\cdots,18} \in M_{7 \times 3}, $ with the coordinates of the landmarks of the 18 macaques, the full Procrustes mean shapes are computed for males and females separately (see Fig. \ref{landmarks.macaque} (b)), and their preshapes, $\{s_i\}_{i=1,\ldots,18}$, (eq. \ref{preshapes}) and sizes $\{x_i\}_{i=1,\ldots,18}$ (eq. \ref{tamañocentroide}), are computed (see Fig. \ref{fig.preshapes} and Table \ref{table.sizes}).

\begin{figure}[htb]
\begin{center}
\begin{tabular}{cc}
\includegraphics[width=7cm]{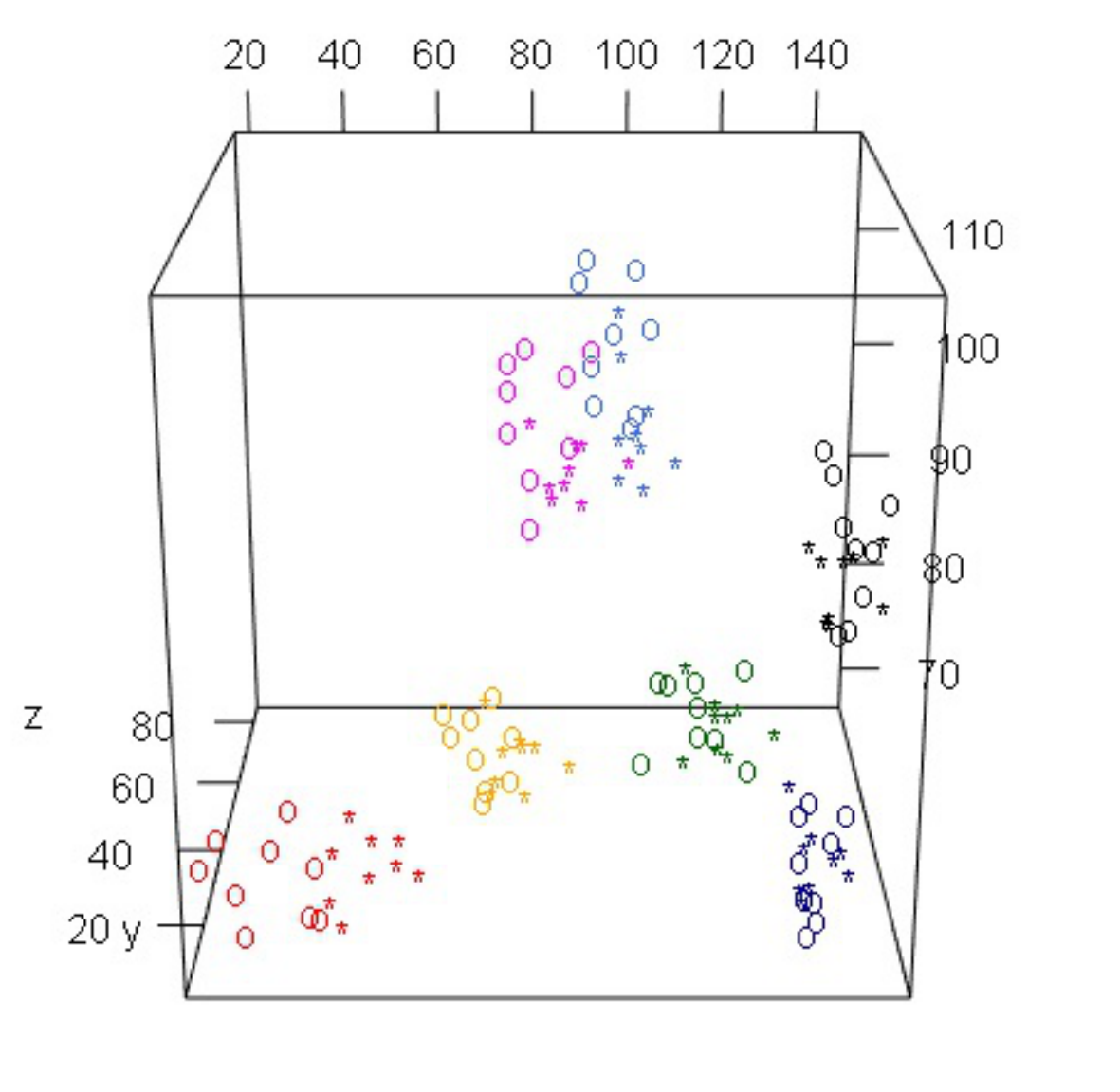} & \includegraphics[width=7cm]{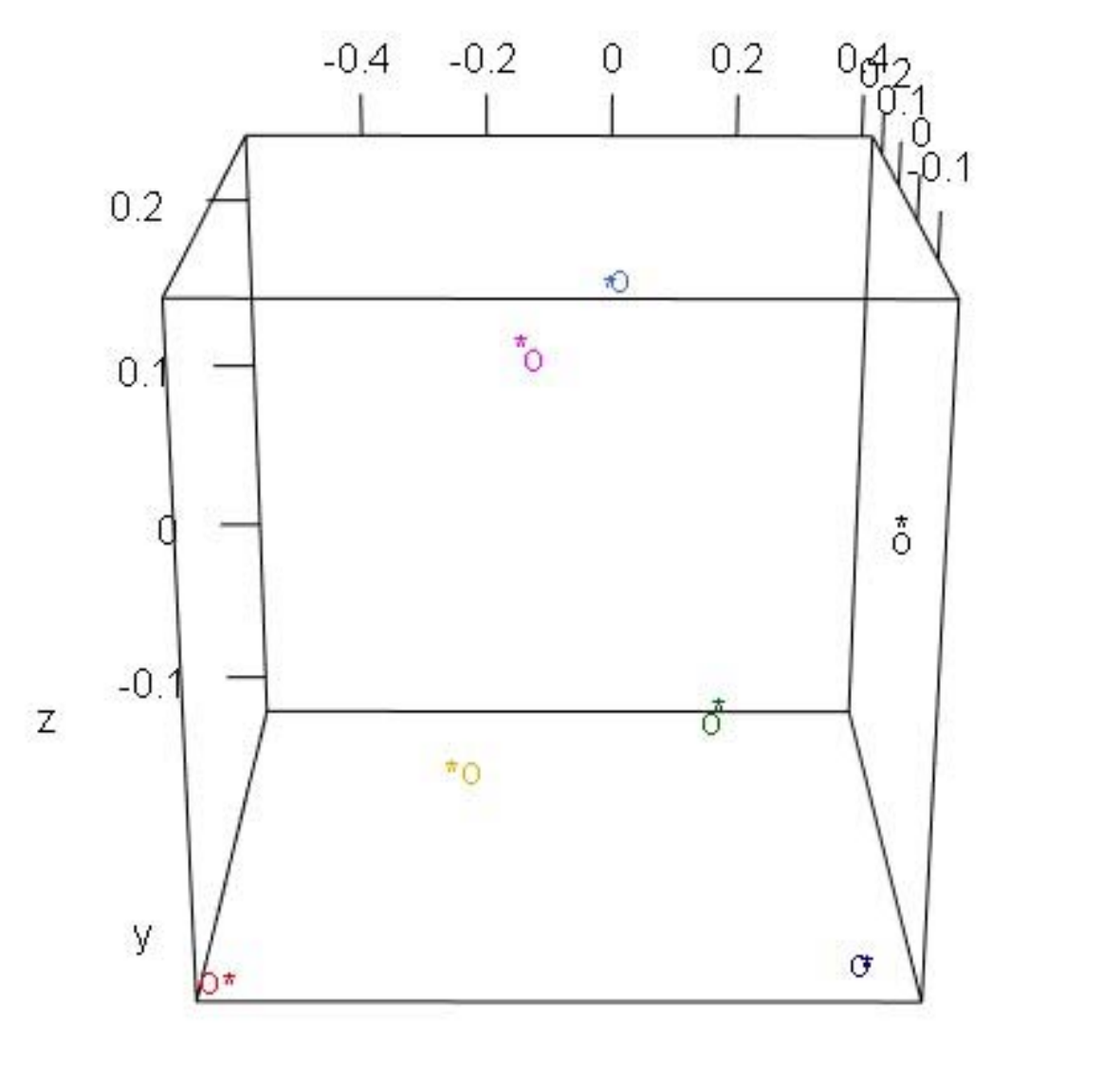}\\
(a) & (b)\\
\end{tabular}
\caption{(a) Landmarks corresponding to the 18 configurations. (b) Landmarks corresponding to the mean shapes of males and females. The different colors represent the different landmarks; the symbol $'o'$ is used for the landmarks of males and the symbol $'*'$ is used for  females. \label{landmarks.macaque}}
\end{center}
\end{figure}

\begin{table}[htb]
\begin{tabular}{cccccccccc}
 \hline
      & m1& m2 & m3 & m4 & m5 & m6 & m7 & m8 & m9 \\
 \hline
 size ($x_i$) &113.9 &104.1 &107.9 &117.6 &113.8 &120.7 &107.4 &117.1 &109.5\\
 \hline
      & f1& f2 & f3 & f4 & f5 & f6 & f7 & f8 & f9 \\
 \hline
 size ($x_i$) &97.1&  87.8& 102.6& 106.6&  94.5 &101.7&  94.4&  97.8& 100.7\\
 \hline
\end{tabular}
 \caption{Sizes of the 18 crania, computed from the landmark configurations. m1, m2,..., m9 denote  the 9 males and f1, f2,..., f9  the 9 females.} \label{table.sizes}
 \end{table}

\begin{figure}[htb]
\begin{center}
\includegraphics[width=7cm]{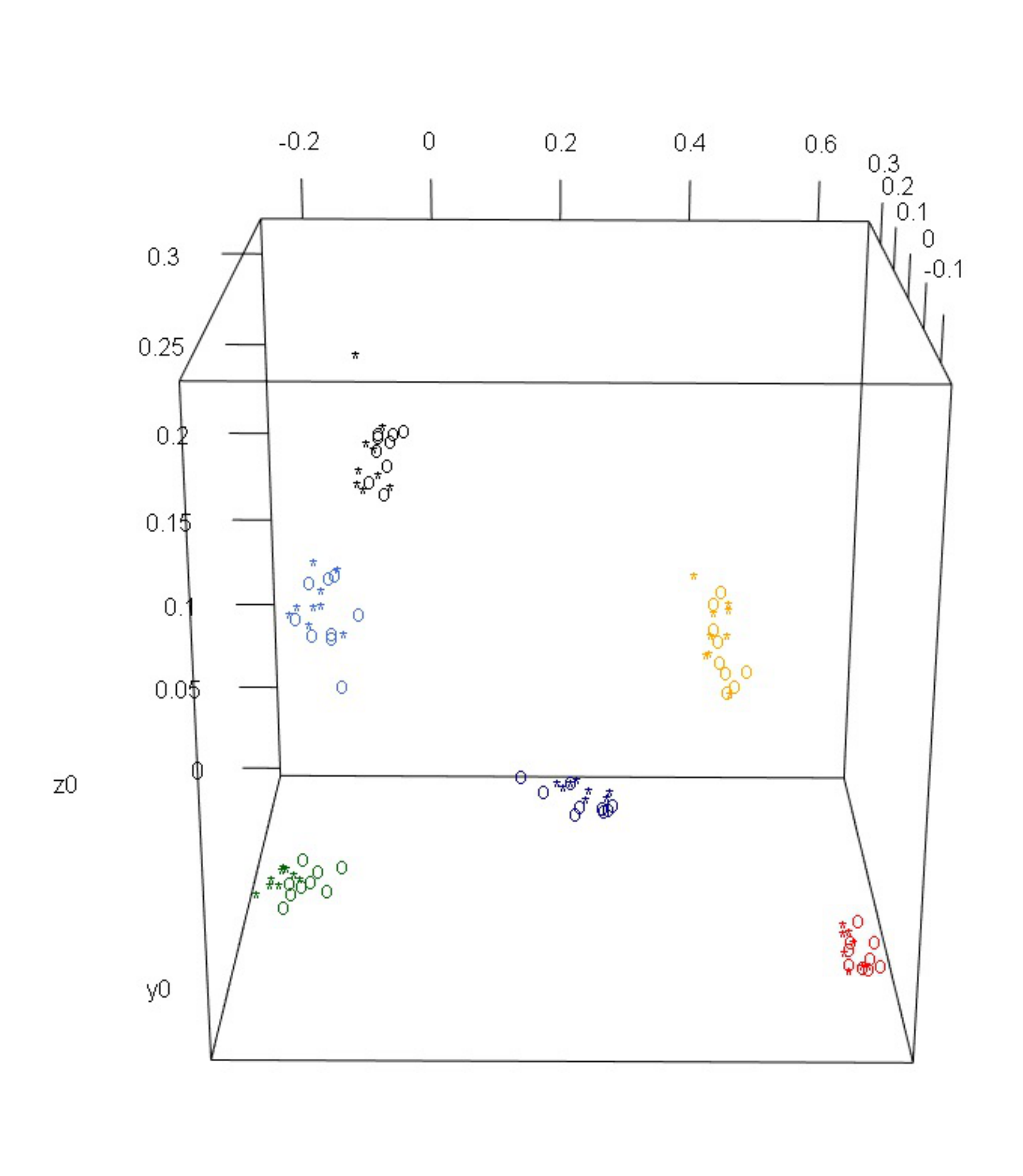}\\
\caption{Pre-shapes of the 18 crania. \label{fig.preshapes}}
\end{center}
\end{figure}

If we define $Y_i=1$ if the $i-th$ cranium belongs to a female and $Y_i=0$ if it belongs to a male, then we can model:
\begin{displaymath}
logit(Pr(Y_i=1))= x_i \beta_1 + g(s_i)+\epsilon_i,   \ \ i=1, \ldots, 18,
\end{displaymath}
and  algorithm  \ref{algorithm_logistic} can be used to fit a logistic partially linear model.

In the smoothing procedure, we considered a Gaussian kernel and  the bandwidth  parameter $h$ was fixed as $h=\pi/100$  by using a  leave-one-out cross-validation (CV) procedure. With this value, a $0\%$ CV error was obtained. With \textit{threshold}=0.0002, the algorithm stops at $7$ iterations.
 (see Table \ref{table.CV}).

\begin{table}
\begin{center}
\begin{tabular}{ccccc}
\hline
$h$                                     & $\pi/100$& $\pi/50$& $\pi/25$ & $\pi/10$\\
\hline
$CV$ ($\%$ of correct classifications)  &      $100 \%$    &     $88.89\%$    &        $88.89\%$    &     $88.89\%$       \\
\hline
\end{tabular}
\end{center}
\caption{Results of the CV analysis to choose the value of the bandwidth parameter for the crania problem.} \label{table.CV}
\end{table}

The estimation procedure provides a $\hat{\beta_1}=-6.02$, signaling that the probability of being  female decreases as  skull size increases, and the values presented in Table
\ref{table.g} were obtained for the nonparametric part of the model.

\begin{table}[htb]
\begin{tabular}{cccccccccc}
 \hline
      & m1& m2 & m3 & m4 & m5 & m6 & m7 & m8 & m9 \\
 \hline
 $\hat{g}(s_i)$ & 655.1 &616.4 &638.8& 682.5& 675.9& 675.4& 635.3& 659.3& 647.4 \\
 \hline
      & f1& f2 & f3 & f4 & f5 & f6 & f7 & f8 & f9 \\
 \hline
 $\hat{g}(s_i)$ &619.8 & 565.5 &628.2 &647.3 &  603& 620.6 & 601.1& 619.3& 632.8\\
 \hline
\end{tabular}
 \caption{Estimation of the nonparametric part of the logistic PLM for each macaque in the data set.} \label{table.g}
 \end{table}

\section{Application to children's body shapes}\label{our_appl}
The aim of this section is to show how the aforementioned algorithm can be used to  predict the goodness of  fit of a given garment size, i.e. small ($Y_i=1$), good fit ($Y_i=2$) or large ($Y_i=3$), as a function of the garment size, the size of the child and his/her shape.

There are multiple ways to choose the most suitable size in a potential online sales application, and all of them depend on the manufacturer.

A randomly selected sample of $739$ Spanish children aged  $3$ to $12$ was scanned using a Vitus Smart 3D body scanner from Human Solutions.
The children were asked to wear a standard white garment in order to standardize the measurements. Several cameras capture images and associated software
provided by the scanner manufacturers detects the brightest points and
uses them to create a triangulation that provides information about the 3D
spatial location of $3075$  points on the body's surface. \\
The 3D scan data are processed to create of posture-harmonized homologous models to obtain a database of individual 3D homologous avatars with one-to-one anatomical  vertex correspondence between them.
As a result, each child's  body shape is represented by $3075$ 3D landmarks.
Because the children's head, hands, legs and feet are not involved in the
shirt size selection, these parts were discarded from the scans, and a
total of 1423 3D landmarks per child were considered, i.e. each child's the body shape  was represented by a $1423 \times 3$ configuration matrix.
Two of them are shown in Figure \ref{fig.landmarks}.

\begin{figure}[htb]
\begin{center}
\begin{tabular}{cc}
\includegraphics[width=7cm]{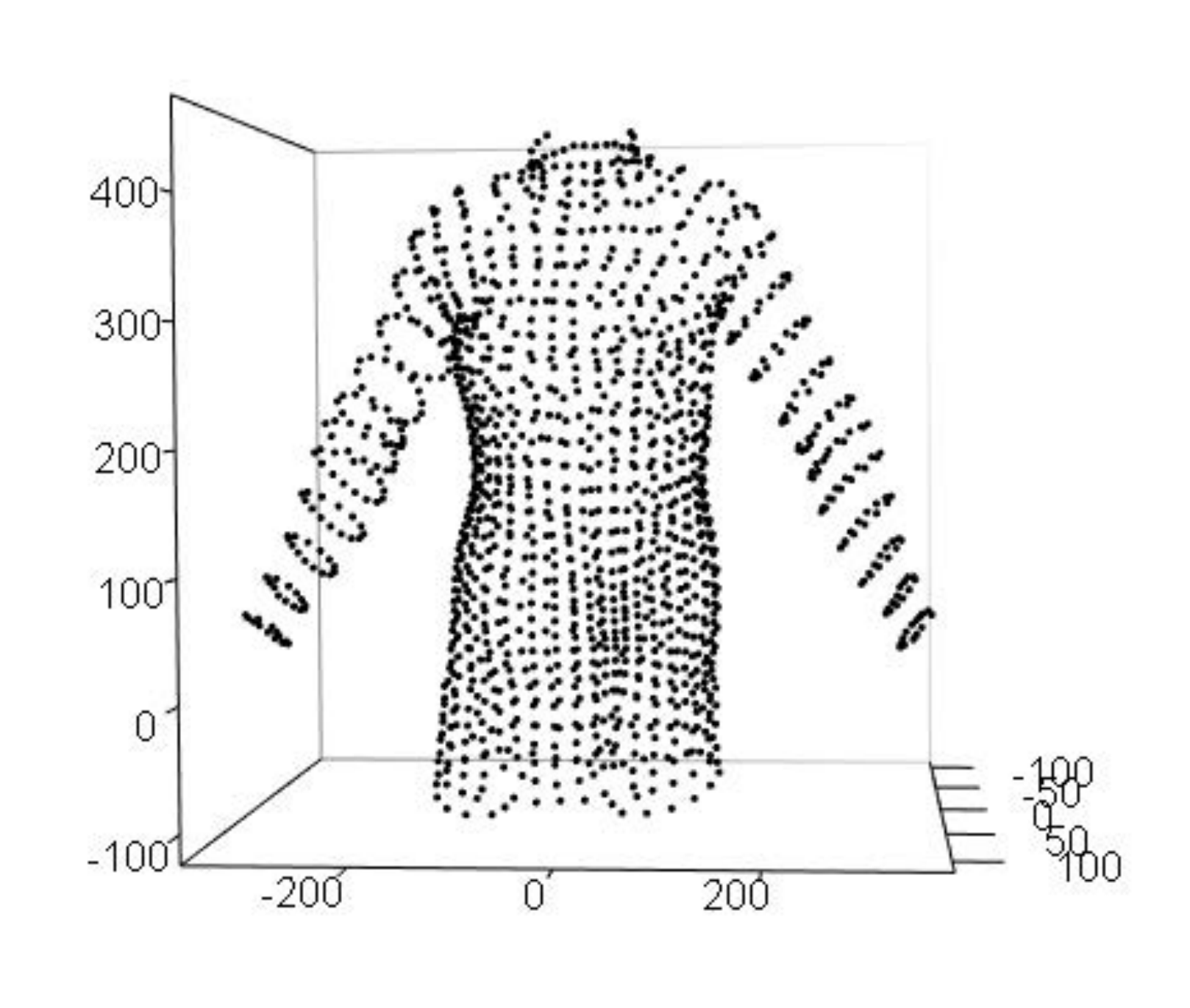} & \includegraphics[width=7cm]{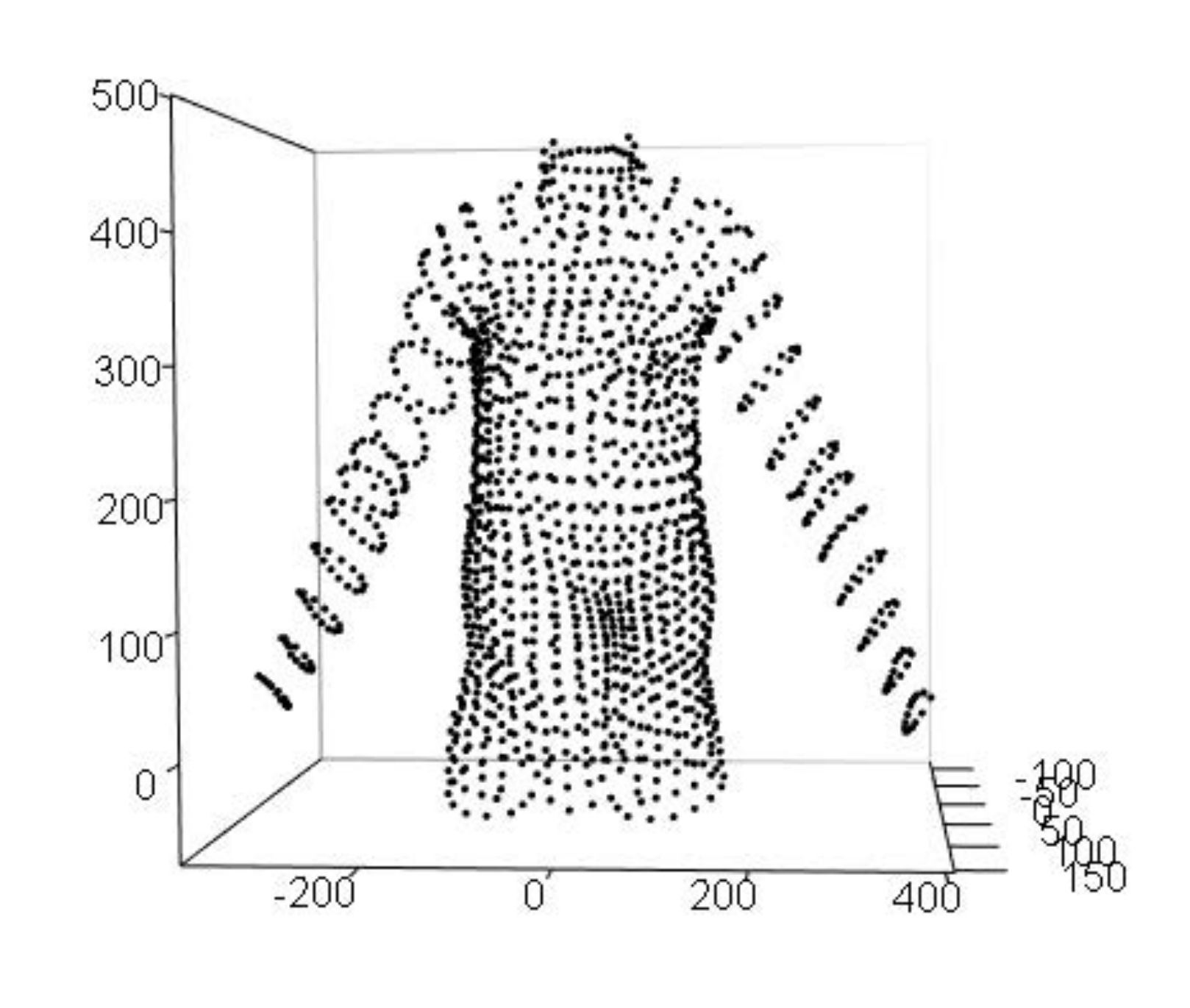}\\
(a) & (b)
\end{tabular}
\caption{3D landmarks of  (a) a girl and (b) a boy in the data set. \label{fig.landmarks}}
\end{center}
\end{figure}

Seventy eight of these children performed an additional fit test. All of them tried on the same shirt model in different sizes: the supposedly correct size, the size above and the size below.  Then, an expert in clothing and design qualitatively evaluated  the  fit in each case (as small, correct fit or large). Due to lack of cooperation by some of the children, not all the children tried on all the three sizes, but only two sizes or even one. In 24 cases, only two sizes were evaluated, and 9 children tested just one shirt size. There were 7 shirt sizes available, supposedly corresponding to ages 3, 4, 5, 6, 8, 10 and 12.

Two subsamples are considered, the sample consisting of the $37$ boys and that consisting of the $41$ girls in the data base.

Algorithm \ref{algorithm_ordinal} is applied to fit the model:
\begin{displaymath}
logit(P(Y_i\leq k | x_{i1},x_{i2},s_i))=\beta_1 x_{i1} + \beta_2 x_{i2} + g_k(s_i)   \ \ i=1, \ldots, n, \ \ k=1,2
\end{displaymath}
and
\begin{displaymath} \label{pglm22}
x_{ij}=\phi_j(s_i) + \eta_{ij},  \ \ i=1, \ldots, n, \ \ j=1,2,
\end{displaymath}
to each subsample,  $s_i$ being the body shape of the $i$-th child,  $x_{i1}$ his/her body size and $x_{i2}$ the size evaluated.

We consider  a Gaussian kernel again and in order to  choose the value of the bandwidth  parameter $h$ and, at the same time, perform a quantitative analysis of the effectiveness of the method, a leave-one-out cross-validation study is conducted. At each step of this study, a child is left out, and  his/her fit predicted for the supposedly correct size, the size above and the size below.  In Table \ref{table.CV2} we can see the percentage of correct classifications in each case.

\begin{table}
\begin{center}
\scalebox{0.8}{
\begin{tabular}{ccccccccccc}
\hline
\multicolumn{2}{c}{$h$}                & $\pi/50$ & $\pi/100$ & $\pi/120$ & $\pi/140$ & $\pi/160$ & $\pi/180$ & $\pi/200$ & $\pi/220$& $\pi/250$ \\
\hline
\multirow{2}{*}{CV (\%) accuracy} &boys     & $52.86$ & $60.00$ & $65.71$ & $64.29 $ & $64.29$ & $68.57$ &  $60.00$&$58.57$ &$52.86$\\
                                  &girls    & $46.57$ & $54.79$ & $71.23$ & $69.86 $ & $71.23$ & $67.12$ &$69.86$& $61.64$& $47.94$\\
\hline
\end{tabular}
}
\end{center}
\caption{Results of the CV analysis for the children's body shape problem.} \label{table.CV2}
\end{table}

The estimation procedure using the full data set provides $\hat{\beta}=(-1.4959, 0.004707)$ with $h=\pi/180$ for boys and $\hat{\beta}=(-1.4158,0.005101)$ with $h=\pi/160$
for girls. With \textit{threshold}=0.0002, the algorithm stops at 783 and 618 iterations respectively.

\subsection{Comparison with other methods}

In \cite{Pierolaetal16}, the authors used ordered logistic regression and random forest methodologies to predict a garment's goodness of  fit from the differences between the measurements of the reference mannequin for the evaluated size and the child's anthropometric
measurements. We could also have  used different children's anthropometric measurements to fit a classical proportional odds model \cite{mccullagh1980regression}. So given the response variable $Y$ with $K=3$ ordered categories, and given $X$ a vector of explicative variables formed by  the garment size to evaluate and the 27 children's anthropometric measurements considered by \cite{Pierolaetal16}, we could have fitted:

\begin{equation}\label{orderedlogistico}
logit[P(Y\leq k\mid x)]=\alpha_k+\beta' x, k=1,2
\end{equation}

Performing a leave-one-out cross-validation study, choosing the
model on each step  by a forward stepwise model selection based on likelihood ratio tests  \cite{paqueteordinalR}, we obtain worse results than those obtained with our methodology. The percentages of correct classifications are now  $61.76\%$ for boys and $66.19\%$ for girls.

On the other hand, as stated in the introduction, the shape space and size-and-shape space are not flat Euclidean spaces, so  classical statistical methods cannot be directly applied to the manifold valued data. However, if the sample has little variability, the problem can be transferred to a tangent space (at the Procrustes mean of these shapes or size-and-shapes, for example) and then standard multivariate procedures can be performed in this space \cite{DrydenMardia16}, such as Principal Component Analysis (PCA). With this approach, in order to reduce the dimensionality of the data set, the first $p$ PC scores, which summarize most of the variability in the tangent plane data, are usually chosen.

The tangent space is defined from a point called pole, so  the distance from the shape to the pole is preserved, i.e. the distance from a point in the manifold to the pole is equal to the Euclidean distance between its projections in
the tangent space. As one moves away from the pole, the Euclidean distances between some pairs of points in the tangent space are smaller than their corresponding shape distances. This distortion becomes larger as one considers points further from it. For this reason, the pole should be taken close to all of the points and the mean of the observed shapes is the best choice \citep{DrydenMardia16}.

So, given the configuration matrices $X_{i} \in M_{1423\times 3}$, the size $s_i$ of each child is obtained and the full Procrustes mean shapes are computed for boys and girls separately. Then, the coordinates of the projection of $X_{i} \in M_{1423\times 3}$ onto the tangent plane defined at their corresponding mean shape are obtained. The first PC scores of these coordinates are calculated and they will be used as covariates in our predictive model. The first PC components that explain $98 \%$ of variability are considered.

So, given the response variable $Y$ with $K=3$ ordered categories and given a vector $X$ with the garment size to evaluate, the child's size and the first PC scores of his/her coordinates in his/her corresponding tangent space, we can fit the model given by Eq. \ref{orderedlogistico}.

Once again, performing a leave-one-out cross-validation study using this model, we obtain worse results than those obtained with our methodology. The percentages of correct classifications are now  $61.43\%$ for boys and $67.12\%$ for girls.

\section{Conclusions }\label{conclusions}
We define GPLMs  on Riemmanian manifolds for the first time. Due the application that we address, our GPLMs have  focused on   Kendall's 3D Shape Space. Although it is an important and common problem in real applications, this problem has not been addressed until now, to the best of our knowledge. We have developed and illustrated the algorithms for estimating the GPLM in two different applications. We have also compared the results with other simpler approximations in the case of the children's garment size matching problem.

Although we have focused on children's shapes in the application, the methodology can also be used to select the right size for adults, men and women. Furthermore, as pointed out in Sect. \ref{introduccion}, the proposed methodologies have great potential in all the fields where statistical shape analysis is used, including biological and medical applications.

Besides opening the door to applications in different fields, other future  work could  focus on other Riemannian manifolds. Moreover, all the work carried out on GPLMs for multivariate data could be extended to the case where variables also take  values on  Riemannian manifolds.

\section*{Acknowledgements}
This paper has been partially supported by the grant
$DPI2013-47279-C2-1-R$ from the Spanish Ministry of Economy and Competitiveness
with FEDER funds and the grant UJI-B2017-13 from Universitat Jaume I. We would also like to
thank the Biomechanics Institute of Valencia for providing us with
the data set.

\appendix

\section{Volume density function in the shape space}\label{Appendix}
In this section we recall the notion of the volume density function. We need to study the volume density function when we work in a curved space. Our first step  is to introduce the definition of the volume density function in its most general sense. This  is attained when the underlying space is a Riemannian manifold $(M,g)$, namely, a smooth manifold $M$ endowed with a metric tensor $g$.  After that, we will particularize it to the explicit formula for the volume density function in the shape space $\Sigma_m^k$, which is the relevant space in this paper. However, since it is easier to work with the pre-shape sphere $S^k_m$ than in shape space $\Sigma_m^k$, our objective will be to make use of the submersion $\pi:S^k_m\to \Sigma^k_m$ to compute  the volume density function in $\Sigma_m^k$ explicitly.

The definition of the volume density function is as follows.
\begin{definition}[see \cite{Henry2009611} for instance]
Let $(M,g)$ be a Riemannian manifold,  let  $s_1\in M$ be a point of $M$, let $T_{s_1}M$ be the tangent space at $s_1$, let $B_r(s_1)$ be an open geodesic ball of radius $r$ centered at $s_1$, let $B_r(0_{s_1})$ be an open ball of radius $r$ centered at $0_{s_1}\in T_{s_1}M$, and let ${\rm inj}(s_1)$ be the injectivity radius at $s_1$, (i.e. the maximum radius $r$ such that the exponential map $\exp:B_r(0_{s_1})\subset T_{s_1}M\to M$ is a diffeomorphism). The \emph{volume density function}  $\theta_{s_1}: B_{{\rm inj}(p)}(p)\to \mathbb{R}_+$ is a function defined for any point $s_2$ of the normal ball  $B_{{\rm inj}(s_1)}(s_1)$ by

\begin{equation}
\theta_{s_1}(s_2)=\frac{\left\vert \det g'\left(\partial_{\overline\psi_i}\vert_w,\, \partial_{\overline\psi_j}\vert_w\right) \right\vert^\frac{1}{2}}{\left\vert \det g''\left(\partial_{\overline\psi_i}\vert_w,\, \partial_{\overline\psi_j}\vert_w\right) \right\vert^\frac{1}{2}}
\end{equation}
where $g'=\exp^*_{s_1}(g)$ is the pullback of $g$ by the exponential map, $g''$ is the canonical metric induced by $g$ in $B_r(0_{s_1})$,  and $(\overline U,\overline \psi)$ is any chart of $B_r(0_{s_1})$ that contains $w= \exp^{-1}(s_2)$.
\end{definition}
In this work, following \cite{Kendalletal09}, we identify a point $Z$ in the pre-shape sphere $S^k_m$ as a matrix. The explicit formula of the volume density function in the shape space $\Sigma_m^k$ is given in the  following theorem.

\begin{theorem}\label{theovolume}
Let $Z_1$ and $Z_2$ be two points of $S^m_k$. Let $\pi: S^m_k\to \Sigma^m_k$ be a Riemannian submersion from $S^m_k$ to $\Sigma^m_k$. Suppose that  ${\rm rank}(Z_1)\geq m-1$ (i.e. $\pi(Z)$ is not a singular point), then the volume density function $  \theta_{\pi(Z_1)}(\pi(Z_2))$ is
\begin{equation*}
  \theta_{\pi(Z_1)}(\pi(Z_2))=\left\{\begin{array}{ccc}
  \left(\frac{\sin\left(\rho(\pi(Z_1),\pi(Z_2))\right)}{\rho(\pi(Z_1),\pi(Z_2)))}\right)^{m(k-1)-2-\frac{m(m-1)}{2}}&\text{ if }& \pi(Z_1)\neq \pi(Z_2)\\
  1& \text{ if }& \pi(Z_1)= \pi(Z_2)
  \end{array}\right.
\end{equation*}
where here $\rho(\pi(Z_1),\pi(Z_2)))$ is the distance in $\Sigma^m_k$ from $\pi(Z_1)$ to $\pi(Z_2)$.
\end{theorem}

\begin{proof}
Since $\pi: S^k_m\to \Sigma_m^k$ is a Riemannian submersion, for any $\pi(p)\in \Sigma_m^k$, the tangent space $T_{\pi(p)}\Sigma_m^k$ can be identified with the horizontal space $\mathcal{H}_p$ of $T_pS^k_m$. Moreover since $d\pi: {\rm ker}(d\pi)^\perp\to TS_m^k$ is an isometry, for any $v,w\in \mathcal{H}_p$, it is the case that $g''(v,w)=g_{\Sigma^k_m}(d\pi(v),d\pi(w))$.  On the other hand, if $q=\exp_p(w)$, then $d\exp_p(w):\mathcal{H}_p\to\mathcal{H}_q$. Therefore, in order to compute $\theta_{\pi(p)}$ we can make use of the restriction of the exponential map $\exp_p\vert_{\mathcal{H}_p}$ to the horizontal space $\mathcal{H}_p$.

In our particular setting, given  $Z\in S_m^k$, the tangent space is
\begin{equation}
T_ZS_m^k=\left\{V\in M(m,k-1)\, :\, {\rm tr}(ZV^t)=0\right\}
\end{equation}
the horizontal space is
\begin{equation}
\mathcal{H}_Z=\left\{ V\in M(m,k-1)\, :\, {\rm tr}(ZV^t)=0\, \text{ and }ZV^t=VZ^t\right\}
\end{equation}
and the exponential map is given by
\begin{equation}
\exp_Z(V)=Z\cos(\Vert V\Vert)+\frac{V}{\Vert V\Vert}\sin(\Vert V\Vert)
\end{equation}

We are now going to obtain $\theta_{\pi(Z_1)}(\pi(Z_2))$  using the properties of the exponential map in $S_m^k$.  Given $Z_1\in S_m^k$, $Z_2=\exp_{Z_1}(W)$ for some $W\in\mathcal{H}_{Z_1}$, suppose $\pi(Z_1)\neq \pi(Z_2)$ and suppose moreover that $\pi(Z_2)$ is in a normal ball of $\pi(Z_1)$.  Let us now choose an orthonormal basis $\{V_i\}_{i=1}^d$ of $\mathcal{H}_{Z_1}$ with $V_1=\frac{W}{\Vert W \Vert}$ (and with $d={\rm dim}(\mathcal{H}_p)$). Then
\begin{equation*}
  \begin{aligned}
    d\exp_{Z_1}(W)(V_i)=&\left.{\frac{d}{dt}}\exp_{Z_1}(W+V_it)\right\vert_{t=0}\\
    =&\left.{\frac{d}{dt}}\left(Z_1\cos(\Vert W+V_it\Vert)+\frac{W+V_it}{\Vert W+V_it\Vert}\sin(\Vert W+V_it\Vert)\right)\right\vert_{t=0}\\
    =&Z_1\sin(\Vert W\Vert)\frac{\langle V_i,W\rangle}{\Vert W\vert}+V_i\frac{\sin(\Vert W \Vert)}{\Vert W \Vert}\\
    &-\frac{W}{\Vert W\Vert^3}\langle V_i,W\rangle\sin(\Vert W\Vert)+\frac{W}{\Vert W\Vert^2}\langle V_i,W\rangle\cos(\Vert W\Vert).
    \end{aligned}
  \end{equation*}
Consequently,
\begin{equation*}
  d\exp_{Z_1}(W)(V_i)=\left\{
  \begin{array}{lcc}
    Z\sin(\Vert W\Vert)+\frac{W}{\Vert W\Vert}\cos(\Vert W\Vert)& \text { if }& i=1\\
    V_i\frac{\sin(\Vert W\Vert)}{\Vert W\Vert}& \text { if }& i\neq 1
  \end{array}
  \right.
  \end{equation*}
Hence, since  $\pi(Z_2)$ is in a normal ball of $\pi(Z_1)$, there is a minimal and horizontal geodesic segment starting at $Z_1$  and ending in $Z_2$ with initial velocity $W$ such that $\rho(\pi(Z_1),\pi(Z_2))=\Vert W\Vert$. Taking into account that since $\langle V_i,W\rangle= \langle V_i,V_j\rangle=0$ for any $i\neq j$ ($i,j>1$) and that $\langle V_i,Z\rangle =\langle W, Z_1\rangle=0$ because $W$ and $V_i$ are tangent vectors to $T_{Z_1}S^m_k$ we conclude that

\begin{equation}\label{eq.8}
  \begin{aligned}
    \theta_{\pi(Z_1)}(\pi(Z_2))=&\sqrt{\det\left(\langle d\exp_{Z_1}(W)(V_i),d\exp_{Z_1}(W)(V_j)\rangle\right)}\\
    =&\left(\frac{\sin{\rho(\pi(Z_1),\pi(Z_2))}}{\rho(\pi(Z_1),\pi(Z_2)))}\right)^{d-1}
    \end{aligned}
\end{equation}
where $d$ is the dimension of $\mathcal{H}_{Z_1}$. Now, we are going to compute the dimension of $\mathcal{H}_{Z}$ for any $Z\in  S^k_m$. Since the tangent space $T_ZS^k_m$ can be decomposed as $T_ZS^k_m=\mathcal{H}_Z\oplus \mathcal{V}_Z$, then
$$
{\rm dim}(\mathcal{H}_Z)={\rm dim}(S^k_m)-{\rm dim}(\mathcal{V}_Z)=m(k-1)-1-{\rm dim}(\mathcal{V}_Z)
$$
where ${\rm dim}(\mathcal{V}_Z)$ is the dimension of the fiber $\pi^{-1}(\pi(Z))$. The dimension of the fiber $\pi^{-1}(\pi(Z))$ depends on the rank of $Z$ (see \cite{Kendalletal09}) but if ${\rm rank}(Z)\geq m-1$ (i.e. it is a non singular point), $\pi^{-1}(\pi(Z))$ is homeomorphic to $SO(m)$
(and hence with dimension $\frac{m(m-1)}{2}$). Therefore,
$$
{\rm dim}(\mathcal{H}_Z)=m(k-1)-1-\frac{m(m-1)}{2}
$$
Then,  using equation (\ref{eq.8}), the theorem follows.
\end{proof}


\end{document}